\documentclass[letterpaper, 10 pt, conference]{ieeeconf}  

\IEEEoverridecommandlockouts                              
\usepackage[utf8]{inputenc}
\usepackage[T1]{fontenc}
\usepackage{type1ec}
\usepackage{stmaryrd}
\usepackage{amsmath}

\let\proof\relax
\let\endproof\relax
\usepackage{amsthm}

\usepackage{amssymb}
\usepackage{dsfont}
\usepackage{mathtools}
\usepackage{thm-restate}
\usepackage{xspace}
\usepackage[textwidth=2cm,textsize=small]{todonotes}

\let\labelindent\relax
\usepackage{enumitem}
\usepackage{hyperref}
\usepackage{cleveref}
\usepackage{stackrel}
\usepackage{MnSymbol} 
\usepackage{todonotes}
\usepackage{graphicx}
\usepackage{etoolbox}
\usepackage{textcase}
\usepackage[bb=boondox]{mathalfa}
\usepackage{verbatim}

\usepackage[contents={}]{background}
\newcommand\BGfrom[1]{%
\AddEverypageHook{%
  \ifnum\value{page}>\numexpr#1-1\relax
    \backgroundsetup{
      contents={Over the limit},
      color=red,
      scale=10,
      opacity=0.5
    }%
  \fi
  \BgMaterial%
  }%
}

\usepackage{amstext} 
\usepackage{array}   
\newcolumntype{C}{>{$}c<{$}} 
\newcolumntype{L}{>{$}l<{$}} 

\hypersetup{
    colorlinks,
    linkcolor={red!50!black},
    citecolor={blue!50!black},
    urlcolor={blue!80!black}
}

\usepackage{tikz}
\usetikzlibrary{positioning}
\usetikzlibrary{decorations.markings}

\makeatletter
\patchcmd{\@sect}{\uppercase}{\MakeTextUppercase}{}{}
\patchcmd{\@sect}{\uppercase}{\MakeTextUppercase}{}{}
\makeatother

\newcommand{\tuple}[1]{\left(#1\right)}

\newcommand{\e}{\varepsilon}

\newcommand{\ignore}[1]{}

\newcommand{\wrt}{w.r.t.}

\newcommand{\cf}{cf.}
\newcommand{\eg}{e.g.}

\newcommand{\aka}{a.k.a.}
\newcommand{\ie}{i.e.}

\newcommand{\rhs}{r.h.s.}
\newcommand{\lhs}{l.h.s.}

\newif\ifstartedinmathmode
\newcommand*{\st}{
  \relax\ifmmode\startedinmathmodetrue\else\startedinmathmodefalse\fi
  \ifstartedinmathmode{\;\cdot\;}\else{s.t.}\fi%
}

\newcommand{\D}{\mathbb D}
\newcommand{\N}{\mathbb N}

\newcommand{\Q}{\mathbb Q}
\newcommand{\R}{\mathbb R}

\renewcommand{\SS}{\mathcal S}
\newcommand{\TT}{\mathcal T}

\let\oldcirc\circ
\let\circ\relax
\DeclareMathOperator{\circ}{\oldcirc\,}

\makeatletter
\providecommand*{\shuffle}{%
  \mathbin{\mathpalette\shuffle@{}}%
}
\newcommand*{\shuffle@}[2]{%
  \sbox0{$#1\vcenter{}$}%
  \kern .15\ht0 
  \rlap{\vrule height .25\ht0 depth 0pt width 2.5\ht0}%
  \raise.1\ht0\hbox to 2.5\ht0{%
    \vrule height 1.75\ht0 depth -.1\ht0 width .17\ht0 %
    \hfill
    \vrule height 1.75\ht0 depth -.1\ht0 width .17\ht0 %
    \hfill
    \vrule height 1.75\ht0 depth -.1\ht0 width .17\ht0 %
  }%
  \kern .15\ht0 
}
\makeatother

\newcommand{\limplies}{\Rightarrow}

\newcommand{\card}[1]{\left|#1\right|}
\newcommand{\length}[2]{\left|#2\right|_{#1}}
\newcommand{\set}[1]{\left\{#1\right\}}
\newcommand{\setof}[2]{\set{#1 \;\middle|\; #2}}

\newcommand{\poly}[2]{#1[#2]}
\newcommand{\polyof}[3]{\poly {#1} {#2 \mid #3}}

\newcommand{\powerseries}[2]{#1[\![#2]\!]}

\newcommand{\series}[2]{#1\langle\!\langle#2\rangle\!\rangle}

\newcommand{\ideal}[1]{\langle#1\rangle}
\newcommand{\idealof}[2]{\ideal{#1 \mid #2}}

\newcommand{\sem}[1]{\left\llbracket#1\right\rrbracket}
\newcommand{\support}[1]{\mathsf{supp}(#1)}

\newcommand{\coefficient}[2]{[#1]#2}
\newcommand{\restrict}[2]{\left.#1\right|_{#2}}
\newcommand{\restrictsmall}[2]{#1|_{#2}}

\newcommand{\deriveleft}[1]{\delta^L_{#1}}
\newcommand{\deriveright}[1]{\delta^R_{#1}}

\let\oldpartial\partial
\renewcommand{\partial}[2]{\oldpartial_{#1}{#2}}

\newcommand{\zero}{\mathbb 0}

\newcolumntype{A}{
 >{$}r<{$}
 @{\extracolsep{0pt}}
 >{${}} l <{$}
} 

\newcolumntype{M}{
 >{$}c<{$}
}

\theoremstyle{plain}
\newtheorem{theorem}{Theorem}
\newtheorem*{theorem*}{Theorem}
\newtheorem{lemma}{Lemma}
\newtheorem{corollary}{Corollary}
\newtheorem{claim}{Claim}
\newtheorem*{claim*}{Claim}

\newtheorem*{proposition*}{Proposition}

\newtheorem*{fact*}{Fact}

\theoremstyle{remark}
\newtheorem{remark}{Remark}

\theoremstyle{definition}

\theoremstyle{problem}

\newenvironment{claimproof}{\begin{proof}[Proof of the claim]}{\end{proof}}

\crefname{equation}{}{}
\crefname{definition}{Definition}{Definitions}
\crefname{claim}{Claim}{Claims}
\crefname{proposition}{Proposition}{Propositions}
\crefname{lemma}{Lemma}{Lemmas}
\crefname{corollary}{Corollary}{Corollaries}
\crefname{theorem}{Theorem}{Theorems}
\crefname{figure}{Fig.}{Figures}
\crefname{fact}{Fact}{Facts}
\crefname{example}{Example}{Examples}
\crefname{remark}{Remark}{Remarks}
\crefname{section}{§}{§§}
\crefname{subsection}{§}{§§}
\crefname{subsubsection}{§}{§§}
\crefname{enumi}{}{}
\crefname{problem}{Problem}{Problems}

\title{\LARGE \bf
Algorithmic analysis of systems with affine input and polynomial state
\thanks{Supported by the ERC grant INFSYS, agreement no.~950398.}
}

\author{Lorenzo Clemente$^{\dag}$
\thanks{$^{\dag}$ Faculty of Mathematics, Informatics, and Mechanics, University of Warsaw, Poland
        {\tt\small clementelorenzo@gmail.com}}%
}

\begin{document}

\maketitle
\thispagestyle{empty}
\pagestyle{empty}

\begin{abstract}
        The goal of this paper is to provide exact and terminating algorithms
        for the formal analysis of deterministic continuous-time control systems with affine input and polynomial state dynamics
        (in short, \emph{polynomial systems}).
        We consider the following semantic properties:
        zeroness and equivalence, input independence, linearity, and analyticity.
        Our approach is based on Chen-Fliess series,
        which provide a unique representation of the dynamics of such systems via their formal generating series.
        
        Our starting point is Fliess' seminal work showing how the semantic properties above
        are mirrored by corresponding combinatorial properties on generating series.
        Next, we observe that the generating series of polynomial systems
        coincide with the class of \emph{shuffle-finite series},
        a nonlinear generalisation of Schützenberger's rational series
        which we have recently studied in the context of automata theory and enumerative combinatorics.
        We exploit and extend recent results in the algorithmic analysis of shuffle-finite series
        (such as zeroness, equivalence, and commutativity)
        to show that the semantic properties above can be decided exactly and in finite time for polynomial systems.
        Some of our analyses rely on a novel technical contribution,
        namely that shuffle-finite series are closed under support restrictions with commutative regular languages,
        a result of independent interest.
\end{abstract}



\section{Introduction}

Chen-Fliess series are applied in control theory
to represent the local input-output behaviour of systems with real analytic state dynamics
and affine dependence on the input (in short, \emph{analytic systems})\-
\cite[Sec.~4.1]{NijmeijerVanDerSchaft:1990}, \cite[Sec.~3]{Isidori:NCS:1995}.
They subsume Volterra series with real analytic kernel functions~\cite[Sec.~3]{Fliess:1981},
in particular (bi)linear systems~\cite{Mohler:1970}.
While there exists a wealth of results on the mathematical aspects of Chen-Fliess series
(exploring topics such as
convergence and continuity~\cite{GrayWang:SCL:2002},
relative degrees~\cite{GrayDuffautEspinosaThistsa:Automatica:2014},
detection~\cite{Gray:CDC:2020},
identification \cite{GrayDuffaurEspinosaHaq:ACC:2022},
optimal control~\cite{EspinosaGrayAvellaneda:CDC:2023},
and many more; see~\cite{GrayFPS_2022} for an extensive treatment),
exact algorithms concerning their analyses are currently lacking.

In this work, we show that certain exact algorithmic analyses are possible.
We focus on the following properties:
\begin{enumerate}[label=(P\arabic*)]
    \item \emph{Equality}:
    \label{P1}
    Are the outputs of two systems equal, for all inputs?
    This is a fundamental problem,
    aiming at identifying systems behaving in the same way as input/output transformers,
    regardless of the syntactic dissimilarities of their presentations.
    As a special case, we also consider \emph{zeroness},
    where one asks whether the output of the system is zero, for all inputs.
    
    \item \label{P2}
    \emph{Input independence}\-
    (\aka~\emph{output invariance}~\cite[Sec.~3.3]{Isidori:NCS:1995}, \cite[Sec.~4.3]{NijmeijerVanDerSchaft:1990}):
    Is the output of the system independent of the concrete value of the input?
    This is a very natural problem,
    aiming at identifying the inputs whose value do not affect the output,
    thus enabling simplifications.
    The problem is nontrivial,
    since the input may affect the state of the system
    and thus output invariance may be the effect of complicated cancellations.
    It is an essential preliminary step to
    disturbance and input-output decoupling~\cite[Ch.~7, 8, 9]{NijmeijerVanDerSchaft:1990}.

    \item \emph{Linearity}:
    \label{P3}
    Is the output of the system a linear function of the input?
    Linearity is a cornerstone concept in control theory.
    Linear systems are easier to analyse and control,
    and the cascade composition of a system followed by a linear one usually results in a system in the same class.
    This problem is nontrivial because the state may still exhibit nonlinear behaviour.

    \item \emph{Analyticity}:
    \label{P4}
    Is the output of the system an analytic function of (the integral of) the input?
    When a system is analytic one can replace complex pathways
    by a direct connection to the output (via an integrator),
    which greatly simplifies the design.
\end{enumerate}

\begin{figure}
    \begin{tabular}{lA}
        \multicolumn{1}{c}{semantic property} & \multicolumn{2}{c}{combinatorial property} \\
        \hline
        \cref{P1}~equality
            & g_1 &= g_2 \\
        \cref{P2}~independence of inputs $J$
            & \support g &\subseteq \Sigma_{\setminus J}^* \\
        \cref{P3}~linearity \wrt~inputs $J$
            & \support g &\subseteq \Sigma_{\setminus J}^* \cdot \Sigma_J \cdot \Sigma_{\setminus J}^* \\
        \cref{P4}~analyticity \wrt~inputs $J$
            & \multicolumn{2}{c}{$g$ is commutative in $\Sigma_J$}
    \end{tabular}
    \caption{Combinatorial characterisations of semantic properties.}
    \label{fig:combinatorial characterisations}
\end{figure}
\noindent
These properties have not been studied before from an algorithmic perspective.
We chose a purely formal presentation
(inspired by the coinductive approach to calculus~\cite{EscardoPavlovic:LICS:1998})
and consider the semantics of a control system as a formal functional $\D^m \to \D$,
where $\D := \powerseries \R t$ is the algebra of formal power series with real coefficients.
In order to develop algorithms, we focus on a class of functionals admitting a finite representation:
We consider systems whose dynamics is polynomial in the state $x \in \D^k$ and affine in the input $u = \tuple{u_1, \dots, u_m} \in \D^m$,
which we call \emph{polynomial systems}:
\begin{align}
    \label{eq:polynomial system}
            x'
                = \sum_{j = 0}^m u_j \cdot p_j(x),
                \quad y = q(x),
                \quad x(0) := x_0 \in \R^k,
\end{align}
where $p_0, \dots, p_m \in \poly \R k^k, q \in \poly \R k$ are ($k$-tuples of) $k$-variate polynomials,
$y \in \D$ is the output, and $u_0 := 1$ by convention.
%
Fliess has shown that one can associate to such a system a \emph{generating series} $g : \Sigma^* \to \R$,
in such a way that the semantic properties \cref{P1,P2,P3,P4} can be characterised by
corresponding combinatorial properties on $g$~\cite{Fliess:MST:1976,Fliess:1981}.
We summarise Fliess' characterisations in~\cref{fig:combinatorial characterisations}
(the notations will formally be defined later).

\subsection{Contributions}

We develop algorithms for the combinatorial properties from~\cref{fig:combinatorial characterisations},
thus deciding the corresponding semantic properties of polynomial systems.
Our contributions follow.
\begin{enumerate}[label=(C\arabic*)]
    \item  \label{item:equality}
    We observe that the generating series of polynomial systems coincide with the class of \emph{shuffle-finite series},
    a generalisation of Schützenberger's rational series
    which we have recently studied from an algorithmic perspective~\cite{Clemente:CONCUR:2024}.
    More precisely, zeroness and equality are decidable for shuffle-finite series
    (with elementary computational complexity)~\cite[Theorem~3]{Clemente:CONCUR:2024},
    which immediately solves problem~\cref{P1} for polynomial systems.

    \item \label{item:commutativity}
    In~\cite[Theorem 7]{Clemente:arXiv:LICS:2025} we show that commutativity is decidable for shuffle-finite series,
    and here we observe that the algorithm generalises to decide commutativity \wrt~a subset of the alphabet,
    thus solving \cref{P4}.
    
    \item \label{item:inclusion queries}
    We show that \emph{support inclusion queries} of the form $\support g \subseteq L$
    where $L \subseteq \Sigma^*$ is a \emph{commutative regular language}
    are decidable for shuffle-finite series.
    Since languages of the form $\Sigma_{\setminus J}^*$ and $\Sigma_{\setminus J}^* \cdot \Sigma_J \cdot \Sigma_{\setminus J}^*$
    are commutative regular, this solves~\cref{P2,P3}.
\end{enumerate}
We consider multi-input single-output systems (MISO);
the framework extends with no difficulty to multiple outputs (MIMO).
The formal approach allows us to avoid convergence considerations
and obtain a clean and concise presentation.
Since polynomial systems are analytic,
such considerations can be added with no mathematical difficulty.
Summarising, we obtain the following result.
\begin{theorem}
    The problems~\cref{P1,P2,P3,P4} are decidable for polynomial systems.
\end{theorem}
We establish both~\cref{item:commutativity} and~\cref{item:inclusion queries} by reduction to the zeroness problem.
To achieve this for~\cref{item:inclusion queries},
we exploit the following closure property,
which constitutes the main technical contribution of the paper.
%
\begin{restatable}{theorem}{thmClosureUnderCommutativeRegularSupportRestrictions}
    \label{thm:closure under commutative regular support restriction}
    Shuffle-finite series are effectively closed
    under support restrictions by commutative regular languages.
\end{restatable}
\noindent
By ``effective'' we mean that there is an algorithm which,
given in input a finite description of a shuffle-finite series $g$ and a commutative regular language $L$,
produces in output a finite description of the restriction of $g$ to $L$.
This is a novel effective closure property for shuffle-finite series,
which is a result of independent interest.


\subsection{Related works}

Shuffle-finite series have \emph{finite Lie rank}~\cite{Fliess:IM:1983,Reutenauer:Lie:1986},
however the converse is not true in general (by a simple cardinality argument,
since series of finite Lie rank are not even finitely presented).
While series of finite Lie rank admit an elegant minimisation result~\cite[Theorem~I.1]{Fliess:IM:1983},
it is not clear whether shuffle-finite series admit an analogous characterisation.
%

The coincidence of generating series of polynomial systems
and those recognised by a subclass of weighted Petri nets
(called \emph{basic parallel processes} in the Petri net literature~\cite{Esparza:FI:1997})
has been noted in~\cite{FoursovHespel:INRIA:2006,FoursovHespel:MTNS:2006} (see also~\cite{GrayHerencia-ZapanaDuffautEspinosaGonzalez:SCL:2009}),
however this connection has not been exploited from an algorithmic perspective.
%

The zeroness problem is related to the \emph{zero dynamics} problem~\cite{Isidori:EJC:2013}
which aims at constructing inputs that cause the system's output to be zero,
studied in~\cite{GrayEbrahimi-FardSchmeding:CISS:2021} for Chen-Fliess series.
Zeroness is much stronger since it requires the output to be zero \emph{for all inputs}.


Beyond the properties~\cref{P1,P2,P3,P4} considered here,
other symmetries have been considered for Fliess operators~\cite{GrayVerriest:arXiv:2025}.
For example, \emph{reversible series} (invariance under reversal of the input word)
give rise to functionals which are invariant under time reversal~\cite[Theorem 3.1]{GrayVerriest:arXiv:2025},
however 1) equivalence of the two properties is left open,
and 2) it is not clear whether reversibility for shuffle-finite series is decidable.

Full proofs can be found in the
appendix.
\section{Preliminaries}

\subsection{Formal series}

An \emph{alphabet} is a finite set of symbols $\Sigma$,
which we call \emph{letters}.
A \emph{word} $w = a_1 \cdots a_n$ is a finite sequence of letters $a_1, \dots, a_n$ from the alphabet;
its \emph{length} is $\length {} w := n$.
The \emph{empty word} is denoted by $\e$ and has length $\length {} \e = 0$.
The set of all words of finite length over $\Sigma$ is denoted by $\Sigma^*$;
together with the empty word it forms a noncommutative monoid under concatenation.
A mapping $f : \Sigma^* \to \R$ is called a \emph{series}.
We denote the set of series by $\series \R \Sigma$.
The \emph{coefficient} of a word $w \in \Sigma^*$ in a series $f \in \series \R \Sigma$ is denoted by $\coefficient w f := f(w)$.
%
The \emph{support} of a series $f \in \series \R \Sigma$, denoted by $\support f$,
is the set of words $w \in \Sigma^*$ such that $\coefficient w f \neq 0$.

The set of series is equipped with several operations.
It carries the structure of a vector space (over $\R$),
with zero $\zero$, scalar multiplication $c \cdot f$ with $c \in \R$, and addition $f + g$ defined element-wise by
$\coefficient w \zero := 0$,
$\coefficient w {(c \cdot f)} := c \cdot \coefficient w f$, and
$\coefficient w {(f + g)} := \coefficient w f + \coefficient w g$,
for every $w \in \Sigma^*$.
This vector space is equipped with two important families of linear maps, called
\emph{left} and \emph{right derivatives}
$\deriveleft a, \deriveright a : \series \R \Sigma \to \series \R \Sigma$, for every $a \in \Sigma$.
They are the series analogues of language quotients:
For every series $f$ 
they are defined by $\coefficient w {\deriveleft a f} := \coefficient {a \cdot w} f$,
resp., $\coefficient w {\deriveright a f} := \coefficient {w \cdot a} f$,
for every $w \in \Sigma^*$.

The \emph{order} of a series $f$ is $\infty$ if $f = 0$,
and otherwise is the least $\length {} w \in \N$ \st~$\coefficient w f \neq 0$.
A family of series $\setof {f_i} {i \in I}$ is \emph{summable} if for every $n \in \N$
there are only finitely many series $f_i$'s of order $\leq n$,
and in this case the series $\sum_{i \in I} f_i$ is well-defined.
Since the family $\setof {f_w \cdot w} {w \in \Sigma^*}$ is summable,
any series $f$ can be written as the possibly infinite sum $f = \sum_{w \in \Sigma^*} f_w \cdot w$.

Finally, we equip the vector space of series with the \emph{shuffle product} operation ``$\shuffle$'',
which turns it into an algebra (all algebras considered in the paper are over $\R$).
Shuffle can be defined first on words,
and then extended to series by continuity (in a suitable topology)~\cite[Sec.~6.3]{Lothaire:CUP:1997}.
We will use a coinductive definition~\cite[Definition 8.1]{BasoldHansenPinRutten:MSCS:2017}.
The \emph{shuffle product} $f \shuffle g$ is the unique series \st
\begin{align}
    \tag{${\shuffle}$-$\e$}
    \label{eq:shuffle:base}
    \coefficient \e {(f \shuffle g)}
        &= f_\e \cdot g_\e, \\
    \tag{${\shuffle}$-$\deriveleft a$}
    \label{eq:shuffle:step}
        \deriveleft a (f \shuffle g)
        &= \deriveleft a f \shuffle g + f \shuffle \deriveleft a g,
        \quad \forall a \in \Sigma.
\end{align}
This is an associative and commutative operation,
with identity the series $1 \cdot \e$.
For instance, $ab \shuffle a = 2 \cdot aab + aba$.
It originates in the work of Eilenberg and MacLane in homological algebra~\cite{EilenbergMacLane:AM:1953},
and was introduced in automata theory by Fliess under the name of \emph{Hurwitz product}~\cite{Fliess:1974}.
It is the series analogue of the shuffle product in language theory,
and it finds applications in concurrency theory, where it models the interleaving semantics of process composition~\cite{Clemente:CONCUR:2024}.

We refer to~\cite{BerstelReutenauer:CUP:2010} for an extensive introduction to formal series,
and to~\cite{CoxLittleOShea:Ideals:2015} for basic notions from algebraic geometry. 

\subsection{Formal functionals}

We consider inputs from $\D := \powerseries \R t$, the set of univariate power series,
which we write in exponential notation as $u = \sum_{n=0}^\infty u_n \cdot \frac {t^n} {n!} \in \D$.
It is an algebra under the usual operations of scalar product, addition, and multiplication.
We denote by $\coefficient {t^n} u$ the coefficient $u_n$.
The notion of order and summability are defined as for series.
The \emph{formal derivative} and the \emph{formal integral} of $u \in \D$ are
\begin{align}
    u' := \sum_{n=0}^\infty u_{n+1} \cdot \frac {t^n} {n!} 
        \text{,\qquad resp.,\quad}
            \int u := \sum_{n=1}^\infty u_{n-1} \cdot \frac {t^n} {n!}. 
\end{align}
Integration increases the order by one.
We have the fundamental relation $(\int v)' = v$
and product rule $(u \cdot v)' = u' \cdot v + u \cdot v'$.
A \emph{formal functional} with $m$ inputs is a mapping $F : \D^m \to \D$.
It is \emph{causal} if $\coefficient {t^n} {(F(u_1, \dots, u_m))}$ depends only on $\coefficient {t^i} {u_j}$ for $0 \leq i < n$ and $1 \leq j \leq m$.
Functionals $F, G$ can be multiplied by scalars $c \cdot F$ ($c \in \R$),
added $F + G$, and multiplied $F \cdot G$,
all operations being defined pointwise,
giving rise to an algebra of functionals.
In the next section, we describe a finite syntax for a class of causal functionals.

\subsection{Polynomial control systems}

For a natural number $k \in \N$,
let $\poly \R k$ denote the algebra of $k$-variate polynomials.
A \emph{polynomial system} is a tuple $\SS = \tuple{m, k, x_0, p_0, \dots, p_m, q}$
where $m, k \in \N$ are the number of inputs, resp., the dimension of the state space,
$x_0 \in \R^k$ is the \emph{initial state},
$p_0, \dots, p_m \in \poly \R k^k$ are tuples of polynomials representing the \emph{state dynamics},
and $q \in \poly \R k$ is a polynomial representing the \emph{output}.
When speaking of a polynomial system in an algorithmic context,
we assume that all data is over the rational numbers $\Q$.
The \emph{semantics} of a polynomial system is the functional $\sem \SS : \D^m \to \D$ which is defined as follows.
Consider a tuple of inputs $u = \tuple{u_1, \dots, u_m} \in \D^m$ (with the convention $u_0 = 1$)
and the system of power series ordinary differential equations~\cref{eq:polynomial system}.
%
%
By the Picard-Lindelöf theorem, it admits a unique solution $x \in \D^k, y \in \D$.
The output to the polynomial system is $\sem \SS(u) := y \in \D$.
The semantics of a polynomial system 
is in fact an analytic functional, in the sense of the next section.

\subsection{Analytic functionals}

Fix an alphabet $\Sigma = \set{a_0, \dots, a_m}$.
The \emph{formal iterated integral} (on $m$ inputs) is the mapping $F : \Sigma^* \to (\D^m \to \D)$
that maps $w \in \Sigma^*$ to the causal functional $F_w : \D^m \to \D$ defined by induction on $w$ as follows.
For every tuple of inputs $u = \tuple{u_1, \dots, u_m} \in \D^m$ (by convention we set $u_0 := 1$),
input symbol $a_j \in \Sigma$, and word $w \in \Sigma^*$,
%
%
\begin{align}
    \label{eq:iterated integral}
    F_\e(u) := 1
        \quad\text{and}\quad
            F_{a_j \cdot w}(u) := \int (u_j \cdot F_w(u)).
\end{align}
Note that the order of $F_w(u)$ is $\geq \length {} w$,
and thus for every $u$ the family of power series $\setof {F_w(u)} {w \in \Sigma^*}$ is summable.
Consequently, we can extend $F$ to the \emph{formal Fliess operator} $F : \series \R \Sigma \to (\D^m \to \D)$
by defining, for every $g \in \series \R \Sigma$, 
\begin{align}
    \label{eq:causal analytic functional}
    \nonumber
    &F_g : \D^m \to \D \\
    &F_g(u) 
        := \sum_{w \in \Sigma^*} g_w \cdot F_w(u) \in \D, \quad \forall u \in \D^m.
\end{align}
(An equivalent treatment based on a notion of composition of series can be found in~\cite{GrayWang:MTNS:2008}).
Causal functionals of the form $F_g$ are called \emph{analytic}.
Notice that $\coefficient {t^0} {(F_g(u))} = g_\e$.
The following classic lemma explains the basic (and beautiful) properties of the Fliess operator.
\begin{restatable}[\protect{\cite{Fliess:1981}}]{lemma}{lemFliessHomomorphism}
    \label{lem:Fliess homomorphism}
    The Fliess operator is a homomorphism from the shuffle algebra of series
    to the algebra of causal functionals.
    In other words, for every series $f, g \in \series \R \Sigma$,
    \begin{align}
        \label{eq:scalar}
        F_{c \cdot f}
            &= c \cdot F_f,
                && \forall c \in \R \\
        \label{eq:sum}
        F_{f + g}
            &= F_f + F_g, \\
        \label{eq:product}
        F_{f \shuffle g}
            &= F_f \cdot F_g.
    \end{align}
\end{restatable}
\noindent
In the next section we tie the knot
by recalling that the semantics of a polynomial system $\SS$ is an analytic functional $\sem \SS = F_g$
for a series $g$ belonging to a well-behaved class.
\section{Shuffle automata and shuffle-finite series}

We recall an automaton-like model recognising series,
which we call shuffle automata~\cite{Clemente:arXiv:LICS:2025}.
%
%
They have previously appeared in~\cite{Clemente:CONCUR:2024} under the name~\emph{weighted basic parallel processes},
highlighting the connection to Petri nets
(the same observation has appeared in~\cite{FoursovHespel:MTNS:2006,FoursovHespel:INRIA:2006}).
They also arise by specialising \emph{differential representations}~\cite{Fliess:IM:1983,Reutenauer:Lie:1986} from formal to polynomial vector fields.

A \emph{shuffle automaton} is a tuple $A = \tuple{\Sigma, X, \alpha_I, O, \Delta}$
where $\Sigma$ is a finite alphabet,
$X = \set{X_1, \dots, X_k}$ is a finite set of \emph{nonterminals},
$\alpha_I : \poly \R X$ is the \emph{initial configuration},
$O : X \to \R$ is the \emph{output function},
and $\Delta : \Sigma \to X \to \poly \R X$ is the \emph{transition function}.
The transition $\Delta_a : X \to \poly \R X$ (for $a \in \Sigma$)
is extended to a unique \emph{derivation} $\Delta_a : \poly \R X \to \poly \R X$
of the polynomial algebra of configurations:
It is the unique linear map \st, for every configurations $\alpha, \beta \in \poly \R X$,
$\Delta_a(\alpha \cdot \beta) = \Delta_a \alpha \cdot \beta + \alpha \cdot \Delta_a \beta$.
The fact that such an extension exists and is unique is a basic fact from differential algebra, \cf~\cite[page 10, point 4]{Kaplansky:DA:1957}.
Transitions from single letters
are extended to all finite input words homomorphically:
For every configuration $\alpha \in \poly \R X$,
input word $w \in \Sigma^*$, and letter $a \in \Sigma$,
we have $\Delta_\e \alpha := \alpha$
and $\Delta_{a \cdot w} \alpha := \Delta_w \Delta_a \alpha$.
In other words, $\Delta_w$ is the \emph{iterated Lie derivative} generated by $\tuple{\Delta_a}_{a \in \Sigma}$.
The \emph{semantics} of a configuration $\alpha \in \poly \R X$
is the series $\sem \alpha \in \series \R \Sigma$ \st
\begin{align}
    \label{eq:semantics}
    \sem \alpha_w &:= O \Delta_w \alpha := (\Delta_w \alpha)(O X_1, \dots, O X_k),
    \ \forall w \in \Sigma^*.
\end{align}
%
%
The series recognised by the shuffle automaton $A$ is $\sem {\alpha_I}$.
%

\begin{restatable}
    [Properties of the semantics~\protect{\cite[Lemma 8 + Lemma 9]{Clemente:CONCUR:2024}}]
    {lemma}{lemPropertiesofSemanticsOfShuffleAutomata}
    \label{lem:properties of semantics - shuffle automata}
    The semantics of a shuffle automaton is a homomorphism from
    the differential algebra of polynomials
    to the differential shuffle algebra of series.
    %
    In other words, $\sem 0 = \zero$, $\sem 1 = 1 \cdot \e$,
    and, for all configurations $\alpha, \beta \in \poly \R X$,
    \begin{align}
        \label{eq:shuffle:sem:scalar-product}
        \sem {c \cdot \alpha}
            &= c \cdot \sem \alpha 
            && \forall c \in \R, \\
        \label{eq:shuffle:sem:sum}
            \sem {\alpha + \beta}
            &= \sem \alpha + \sem \beta, \\
        \label{eq:shuffle:sem:product}
            \sem {\alpha \cdot \beta}
            &= \sem \alpha \shuffle \sem \beta, \\
        \label{eq:shuffle:sem:derivation}
        \sem {\Delta_a \alpha}
            &= \deriveleft a \sem \alpha
            && \forall a \in \Sigma.
    \end{align}
\end{restatable}

Shuffle automata are finite data structures representing series,
and are thus suitable as inputs to algorithms.
The same class of series admits a semantic presentation,
which we find more convenient in proofs, and which we now recall.
For \emph{generators} $f^{(1)}, \dots, f^{(k)} \in \series \R \Sigma$
let $\poly \R {f^{(1)}, \dots, f^{(k)}}_{\shuffle}$ be the smallest shuffle algebra of series containing the generators.
Algebras of this form are called \emph{finitely generated}.
A series $f \in \series \R \Sigma$ is \emph{shuffle finite}
if it belongs to a finitely generated shuffle algebra closed under left derivatives.
The following equivalent characterisation will constitute our working definition.
\begin{lemma}[\protect{\cf~\cite[Theorem 12]{Clemente:CONCUR:2024}}]
    \label{lem:working def of shuffle finite}
    A series $f \in \series \R \Sigma$ is shuffle finite
    iff there are series $f^{(1)}, \dots, f^{(k)} \in \series \R \Sigma$ \st\-
    \begin{enumerate}
        \item $f \in \poly \R {f^{(1)}, \dots, f^{(k)}}_{\shuffle}$, and
        \item $\deriveleft a f^{(i)} \in \poly \R {f^{(1)}, \dots, f^{(k)}}_{\shuffle}$
        for all $a \in \Sigma$ and $1 \leq i \leq k$.
    \end{enumerate}
\end{lemma}

We have the following coincidence result.
\begin{restatable}[\protect{\cite{Fliess:1981,Clemente:CONCUR:2024}}]{theorem}{thmCoincidence}
    \label{thm:coincidence}
    The semantics of every polynomial system is a causal analytic functional.
    Moreover, the following three classes of series coincide:
    \begin{enumerate}
        \item generating series of polynomial systems,
        \item series recognised by shuffle automata, and
        \item shuffle-finite series.
    \end{enumerate}
\end{restatable}
\noindent
The first part of the theorem and equivalence of the first two points
follows from Fliess' fundamental formula~\cite[Theorem~III.2]{Fliess:1981},
which can be obtained from~\cref{eq:causal analytic functional,eq:semantics}.
%
%
Equivalence of the last two points is from~\cite[Theorem 12]{Clemente:CONCUR:2024}.

We recall some closure properties of the class of shuffle-finite series.
They are \emph{effective} in the sense that given shuffle automata recognising the input series
one can construct a shuffle automaton recognising the output series.
They all follow quite easily from~\cref{lem:properties of semantics - shuffle automata} (\cf~\cite[Lemma 10]{Clemente:CONCUR:2024}),
with the exception of closure under right derivative $\deriveright a f$ which is nontrivial~\cite[Lemma 12]{Clemente:arXiv:LICS:2025}.
\begin{lemma}[\protect{\cite{Clemente:CONCUR:2024,Clemente:arXiv:LICS:2025}}]
    \label{lem:shuffle-finite closure properties}
    If $f, g \in \series \R \Sigma$ are shuffle finite then so are
    $c \cdot f$ ($\forall c \in \R$), $f + g$, $f \shuffle g$, $\deriveleft a f$ and $\deriveright a f$ ($\forall a \in \Sigma$).
\end{lemma}

\subsection{Commutative regular support restrictions}

In this section we show a novel closure property of shuffle-finite series.
The \emph{restriction} of a series $f \in \series \R \Sigma$ to a language $L \subseteq \Sigma^*$
is the series $\restrict f L \in \series \R \Sigma$
which agrees with $f$ on $L$, and is zero elsewhere.
Formally, for every $w \in \Sigma^*$,
\begin{align}
    \coefficient w {(\restrict f L)} :=
    \begin{cases}
        \coefficient w f & \text{if $w \in L$}, \\
        0 & \text{otherwise}.
    \end{cases}
\end{align}
A language $L \subseteq \Sigma^*$ is \emph{commutative}
if it is invariant under permutation of letter positions.
%
For example, the language $\set{ab, ba}$ is commutative, while $\set{ab}$ is not.
A language $L$ is \emph{regular} if it is recognised by a finite automaton~\cite[Ch.~2]{HopcroftMotwaniUllman:2000}.
For example, $a^*b^* = \setof{a^mb^n} {m, n \in \N}$ is regular,
while $\setof{a^nb^n}{n \in \N}$ is not.
A \emph{commutative regular language} is a language which is both commutative and regular.
It is well-known that they can be described with Boolean combinations
of basic threshold and modulo constraints on the number of occurrences of each letter~\cite[Ch.~10, Prop.~5.11]{HandookOfFormalLanguages:Vol1:1997}.
For instance
``even number of $a$'s and at most two $b$'s'' describes a commutative regular language.
The main technical result of the paper is the following closure property.

\thmClosureUnderCommutativeRegularSupportRestrictions*
\begin{proof}
    In the proof, we find it convenient to work with another, equivalent, definition of commutative regular languages.
    A language $L \subseteq \Sigma^*$ is \emph{recognisable}
    if there is a finite monoid $\tuple {M, \cdot, 1}$, a subset $F \subseteq M$,
    and a monoid homomorphism $h : \Sigma^* \to M$ \st~$L = h^{-1} F$~\cite[Ch.~III, Sec.~12]{Eilenberg:1974:A}.
    A monoid is \emph{commutative} if $x \cdot y = y \cdot x$ for all $x, y \in M$.
    The classes of commutative regular languages
    and those recognised by commutative monoids coincide.
    (This is a consequence of \emph{Eilenberg's variety theorem}~\cite[Ch.~7, Theorem~3.2]{Eilenberg:1974:B}.)
    Consider a shuffle-finite series $f \in \series \R \Sigma$
    and a commutative regular language $L \subseteq \Sigma^*$
    recognised by the finite commutative monoid $\tuple {M, \cdot, 1}$
    with accepting set $F \subseteq M$ and homomorphism $h : \Sigma^* \to M$.
    By definition, $f \in A$ for a finitely generated shuffle algebra $A := \poly \R {f^{(1)}, \dots, f^{(k)}}_{\shuffle}$
    closed under left derivatives.
    For every series $g \in \series \R \Sigma$ and $m \in M$,
    let $\restrict g m$ be the restriction of $g$ to the language $h^{-1} m$.
    Since $\setof {h^{-1} m \subseteq \Sigma^*} {m \in M}$ is a finite partition of $\Sigma^*$,
    we can write $g$ as a finite sum of its restrictions:
    \begin{align}
        \label{eq:decomposition}
        g = \sum_{m \in M} \restrict g m.
    \end{align}
    %
    %
    Thanks to this decomposition and~\cref{lem:shuffle-finite closure properties},
    it suffices to show that $\restrict f m$ is shuffle finite, for arbitrary $m \in M$.
    We begin with some observations on the interaction of restriction
    and the basic series operations.
    \begin{restatable}{claim*}{claimRegularSupportRestrictions}
        For every series $f, g \in \series \R \Sigma$ and $m \in M$,
        \begin{align}
            \label{eq:claim:restrict}
            \restrict {(\restrict f {m'})} m &=
                \begin{cases}
                    \restrict f m & \text{if } m = m', \\
                    0 & \text{otherwise}.
                \end{cases} \\
            \label{eq:claim:const}
            \restrict {(c \cdot f)} m
                &= c \cdot \restrict f m,
                && \forall c \in \R \\
            \label{eq:claim:sum}
            \restrict {(f + g)} m
                &= \restrict f m + \restrict g m, \\
            \label{eq:claim:shuffle}
            \restrict {(f \shuffle g)} m
                &= \sum_{m = x \cdot y} (\restrict f x \shuffle \restrict g y), \\
            \label{eq:claim:derive}
            \deriveleft a {(\restrict f m)}
                &= \sum_{m = h(a) \cdot m'} \restrict {\left(\deriveleft a f\right)} {m'},
                && \forall a \in \Sigma
        \end{align}
        where the sum in~\cref{eq:claim:shuffle} ranges over $x, y \in M$,
        and that in \cref{eq:claim:derive} over $m' \in M$.
    \end{restatable}
    Consider the shuffle algebra generated by all restrictions of the original generators,
    $B := \polyof \R {\restrict {f^{(i)}} m} {1 \leq i \leq k, m \in M}_{\shuffle}$.
    By the decomposition~\cref{eq:decomposition}, we have $A \subseteq B$ and thus $f \in B$.
    We first show that $B$ is closed under restrictions.
    \begin{claim*}
        For every series $g \in B$ and $m \in M$,
        we have $\restrict g m \in B$.
    \end{claim*}
    \begin{claimproof}
        Since $g \in B$, there is a polynomial $p \in \poly \R {k \cdot \card M}$
        \st~$g = p((\restrict {f^{(i)}} m)_{1 \leq i \leq k, m \in M})$.
        Then $\restrict g m \in B$ is proved by an induction on the structure of $p$,
        where each case is handled by~\cref{eq:claim:restrict,eq:claim:const,eq:claim:sum,eq:claim:shuffle}.
    \end{claimproof}
    \noindent
    %
    Consider $g := \deriveleft a {(\restrict {f^{(i)}} m)}$.
    In order to conclude the proof, by~\cref{lem:working def of shuffle finite}
    it suffices to show $g \in B$.
    By~\cref{eq:claim:derive},
    we have $g = \sum_{m = m' \cdot h(a)} \restrict {\left(\deriveleft a f^{(i)}\right)} {m'} $,
    where the sum is over $m' \in M$.
    But $\deriveleft a f^{(i)} \in A \subseteq B$,
    and thus by the last claim $g \in B$.
\end{proof}






\section{Decision problems for shuffle-finite series}
\label{sec:decision problems for shuffle-finite series}
We consider the following three decision problems for shuffle-finite series:
equality (and zeroness), regular support inclusion, and commutativity.
In each case, shuffle-finite series are finitely presented via shuffle automata.
For the sake of computability, all quantities are rational numbers $\Q$, expressed in binary notation.

\subsection{Equality and zeroness problems}

In the \emph{equality problem} we are given two shuffle-finite series $f, g \in \series \Q \Sigma$ and we ask whether $f = g$.
In the special case $g = 0$, we obtain the \emph{zeroness} problem.
%
%
\begin{theorem}[\protect{\cite[Theorem 1]{Clemente:CONCUR:2024}}]
    \label{thm:shuffle-finite zeroness}
    The equality and zeroness problems for shuffle-finite series are decidable.
\end{theorem}
\noindent
We quickly recall the algorithm from~\cite{Clemente:CONCUR:2024}.
Thanks to the effective closure properties of shuffle-finite series (\cf~\cref{lem:properties of semantics - shuffle automata}),
equality reduces to zeroness.
Let $f$ be a shuffle-finite series recognised by a shuffle automaton $A = \tuple{\Sigma, X, \alpha_I, O, \Delta}$.
The zeroness algorithm constructs nondecreasing chain of polynomial ideals
\begin{align*}
    I_0 \subseteq I_1 \subseteq \cdots \subseteq \poly \R X,
    \quad \text{where } I_n := \idealof {\Delta_w \alpha_I}{w \in \Sigma^{\leq n}}.
\end{align*}
Here, $I_n$ is the ideal generated by all configurations reachable from the initial configuration $\alpha_I$ by reading words of length $\leq n$.
The chain above terminates at some $I_N = I_{N+1} = \cdots$
by Hilbert's finite basis theorem~\cite[Theorem 4, §5, Ch. 2]{CoxLittleOShea:Ideals:2015},
and, generalising the analysis of~\cite[Theorem 4]{NovikovYakovenko:1999},
one shows that $N$ is at most doubly exponential in the size of the input (\cf~\cite{Clemente:WBPP:arXiv:2024} for more details).
Using algorithms for ideal equality~\cite{Mayr:STACS:1989} (\eg, via Gr\"obner bases),
$N$ can be computed and zeroness reduces to checking whether the generators of $I_N$ vanish on the output function $O$.

\subsection{Regular support inclusion problem}

In the \emph{regular support inclusion problem} we are given
a shuffle-finite series $f \in \series \Q \Sigma$ and a regular language $L \subseteq \Sigma^*$,
and we ask whether $\support f \subseteq L$.
This problem is decidable when $L$ is a commutative regular language.

\begin{theorem}
    \label{thm:regular support inclusion}
    The commutative regular support inclusion problem for shuffle-finite series is decidable.
\end{theorem}
\begin{proof}
    The query $\support g \subseteq L$ reduces to $h := \restrict g {\Sigma^* \setminus L} = 0$.
    Since commutative regular languages are closed under complementation~\cite[Theorem 4.5]{HopcroftMotwaniUllman:2000},
    $\Sigma^* \setminus L$ is also commutative regular.
    By \cref{thm:closure under commutative regular support restriction},
    $h$ is effectively shuffle finite, and thus we conclude by \cref{thm:shuffle-finite zeroness}.
\end{proof}

\subsection{Commutativity problem}

Consider a sub-alphabet $\Gamma \subseteq \Sigma$.
For two words $u, v \in \Sigma^*$ we write $u \sim_\Gamma v$
if one word can be obtained from the other by permuting the letters in $\Gamma$.
For instance, for $\Gamma = \set{a_0}$ we have $a_0 a_1 a_2 \sim_\Gamma a_1 a_0 a_2 \sim_\Gamma a_1 a_2 a_0$
but $a_1 a_2 \not\sim_\Gamma a_2 a_1$.
A series $g \in \series \R \Sigma$ is \emph{commutative in $\Gamma$}
if, for every words $u \sim_\Gamma v$ we have $g_u = g_v$.
%
When $\Gamma = \Sigma$, we just say that $g$ is \emph{commutative}.
In the \emph{commutativity problem} we are given a shuffle-finite series $f \in \series \Q \Sigma$ and a sub-alphabet $\Gamma \subseteq \Sigma$,
and we ask whether $f$ is commutative in $\Gamma$.

\begin{theorem}
    \label{thm:commutativity}
    The commutativity problem for shuffle-finite series is decidable.
\end{theorem}
\begin{proof}
    We have shown in~\cite[Theorem 7]{Clemente:arXiv:LICS:2025}
    that the commutativity problem when $\Gamma = \Sigma$ is decidable.
    We extend the decision procedure to any sub-alphabet $\Gamma \subseteq \Sigma$.
    To this end, observe that a series $g \in \series \R \Sigma$ is commutative in $\Gamma$ if, and only if,
    it satisfies the following two family of equations:
    \begin{align}
        \deriveleft a {\deriveleft b g} = \deriveleft b {\deriveleft a g}
            \quad\text{and}\quad
                \deriveleft a g = \deriveright a g,
        &&\forall a, b \in \Gamma.
    \end{align}
    This holds since swaps and rotations suffice to generate all $\sim_\Gamma$-equivalent words.
    It follows that we can decide commutativity in $\Gamma$
    by checking the above equations for all pairs of letters in $\Gamma$,
    which is decidable by~\cref{thm:shuffle-finite zeroness}.
\end{proof}
\section{Formal analysis of polynomial systems}

In this section we leverage the decidability results for shuffle-finite series from~\cref{sec:decision problems for shuffle-finite series}
to decide the semantic properties~\cref{P1,P2,P3,P4} of polynomial systems.

\subsection{Equivalence}

Two functionals $F, G : \D^m \to \D$ are \emph{equal}
if, for every input $u \in \D^m$, we have $F(u) = G(u)$.
Equal causal analytic functionals have the same generating series.
\begin{lemma}[\protect{\cite[Lemme~II.1]{Fliess:1981}}]
    \label{lem:equality}
    Consider two causal analytic functionals $F_g, F_h$.
    Then, $F_g = F_h$ iff $g = h$.
\end{lemma}
\noindent
This is a crucial result in the development of the theory of Chen-Fliess series,
as testified by its repeated occurrence in the literature\-
\cite[Corollaire~II.5]{Fliess:EMS:1986}, \cite[Lemma 2.1]{WangSontag:1989}, \cite[Lemma~3.1.2]{Isidori:NCS:1995}.
See~\cite[Theorem 7]{GrayWang:MTNS:2008} for a proof in the formal setting (\cf~\cite[Theorem 3.38]{GrayFPS_2022}).

The \emph{equivalence problem} for polynomial systems takes as input two polynomial systems $\SS, \TT$
and amounts to decide whether their causal functionals are equal $\sem \SS = \sem \TT$.
%
%
By \cref{thm:coincidence}, the semantics of $\SS, \TT$ are analytic functionals $\sem \SS = F_g$, resp., $\sem \TT = F_h$
for shuffle-finite generating series $g, h \in \series \R \Sigma$.
Thanks to~\cref{lem:equality}, equivalence is thus reduced to deciding $g = h$,
which is an instance of the equality problem for shuffle-finite series.
By applying~\cref{thm:shuffle-finite zeroness}, we obtain the main result of this section.
\begin{theorem}
    The equivalence problem for polynomial systems is decidable. 
\end{theorem}

\begin{remark}[Equivalence vs.~structural equivalence]
    In the equivalence problem the initial states of $\SS, \TT$ are given as part of the input. 
    One could consider a stronger \emph{structural equivalence}
    where $\SS, \TT$ are required to have the same semantics \emph{for every initial state}.
    For instance, the equations $x' = 0, y = x$ represent a system which is equivalent to $0$ for the initial state $x(0) = 0$,
    but not structurally equivalent to it.
    Structural equivalence and zeroness are much simpler problems for polynomial systems.
    For instance, a system is structurally zero iff the output polynomial is identically zero.
\end{remark}
\subsection{Input independence}

Consider a set of indices $J \subseteq \set{1, \dots, m}$.
A functional $F : \D^m \to \D$ is \emph{independent of inputs $J$}
if the value of $F$ does not depend on the inputs $u_j$ for $j \in J$.
We now provide a formal definition.
For a tuple of inputs $u \in \D^m$, let $u_J \in \D^m$
be the tuple of inputs which agrees with $u$ on $J$, and is zero otherwise.
Write $u_{\setminus J}$ for $u_{\set {1, \dots, m} \setminus J}$,
so that we have $u = u_J + u_{\setminus J}$.
A functional $F : \D^m \to \D$ is \emph{independent of inputs $J$} if
\begin{align}
    \label{eq:input independence}
    F(u_{\setminus J} + u_J) = F(u_{\setminus J} + v_J),
    \quad \forall u, v \in \D^m.
\end{align}
For causal analytic functionals input independence can be characterised as a property of the generating series.
For an alphabet $\Sigma = \set{a_0, \dots, a_m}$,
let $\Sigma_J := \setof {a_j \in \Sigma} {j \in J}$
and $\Sigma_{\setminus J} := \Sigma \setminus \Sigma_J$.
\begin{lemma}
    \label{lem:input independence}
    A causal analytic functional $F_g$
    is independent of inputs $J \subseteq \set{1, \dots, m}$ iff
    \begin{align}
        \label{eq:input independence characterisation}
        \support g \subseteq \Sigma_{\setminus J}^*.
    \end{align}
\end{lemma}
\begin{proof}
    If~\cref{eq:input independence characterisation} holds,
    then input independence~\cref{eq:input independence} is obvious
    from the definition of $F_g$~(\cf~\cref{eq:iterated integral,eq:causal analytic functional}).
    On the other hand, assume that input independence~\cref{eq:input independence} holds.
    We adapt the proof from~\cite[Lemma 3.3.1]{Isidori:NCS:1995}.
    Take $v_J = 0$.
    The difference of the two sides of~\cref{eq:input independence} is
    \begin{align*}
        &F_g(u) - F_g(u_{\setminus J})
        = \sum_{w \in \Sigma^*} g_w \cdot (I_w(u) - I_w(u_{\setminus J})) = \\
        &= \sum_{w \in \Sigma^* \cdot \Sigma_J \cdot \Sigma^*} g_w \cdot I_w(u) = F_{\restrict g {\Sigma^* \cdot \Sigma_J \cdot \Sigma^*}}(u) = 0.
    \end{align*}
    We have used $I_w(u) = I_w(u_{\setminus J})$ if $w$ does not contain any letter from $\Sigma_J$,
    and $I_w(u_{\setminus J}) = 0$ otherwise.
    By~\cref{lem:equality}, we get $\restrict g {\Sigma^* \cdot \Sigma_J \cdot \Sigma^*} = 0$,
    from which~\cref{eq:input independence characterisation} follows.
\end{proof}

A polynomial system $\SS$ is \emph{independent of inputs $J$} if its semantics $\sem \SS$ is.
The \emph{input independence problem} takes a polynomial system $\SS$ and a set of inputs $J$ thereof
and asks whether $\SS$ is independent of inputs $J$.
\begin{remark}
    One can find in \cite[Sec.~3.3]{Isidori:NCS:1995} a stronger notion of input independence,
    which requires that the output of the system is independent of the inputs
    \emph{for every initial state}.
    We could call this property \emph{structural independence}.
    This is much stronger than independence~\cref{eq:input independence},
    for instance $x' = u_1 \cdot x, y = x$, which can be solved explicitly as $y = x(0) \cdot e^{u_1}$,
    is independent of $u_1$ for the initial state $x(0) = 0$,
    but not structurally independent.
\end{remark}
Thanks to \cref{thm:coincidence} and \cref{lem:input independence},
we reduce the input independence problem for a polynomial system $\SS$
with shuffle-finite generating series $g$ to the regular support inclusion query \cref{eq:input independence characterisation},
where $\Sigma_J^*$ is commutative regular.
By~\cref{thm:regular support inclusion} we obtain the main result of this section.
\begin{theorem}
    The input independence problem for polynomial systems is decidable. 
\end{theorem}

\subsection{Linearity}

Let $J \subseteq \set {1, \dots, m}$ be a set of indices.
A functional $F$ is \emph{additive on inputs $J$} if,
for every input $u, v \in \D^m$,
\begin{align}
    \label{eq:additivity}
    \begin{aligned}
        F(u_{\setminus J} + u_J + v_J)
        = F(u_{\setminus J} + u_J) + F(u_{\setminus J} + v_J),
    \end{aligned}
\end{align}
it is \emph{proportional on inputs $J$} (\ie, homogeneous of the first degree) if,
for every input $u \in \D^m$ and constant $\alpha \in \R$,
\begin{align}
    \label{eq:proportionality}
    \begin{aligned}
        F(u_{\setminus J} + \alpha \cdot u_J)
        = \alpha \cdot F(u_{\setminus J} + u_J),
    \end{aligned}
\end{align}
and it is \emph{linear on inputs $J$} if it is both additive and proportional on inputs $J$.
%
%
\begin{lemma}[\protect{\cf~\cite[Proposition~II.8]{Fliess:1981}}]
    \label{lem:linearity}
    The following conditions are equivalent for a causal analytic functional $F_g$
    and a set of inputs $J \subseteq \set {1, \dots, m}$.
    \begin{enumerate}[label=(\arabic*)]
        \item $F_g$ is additive on inputs $J$.
        \item $F_g$ is proportional on inputs $J$.
        \item $F_g$ is linear on inputs $J$.
        \item The support of $g$ satisfies
        \begin{align}
            \label{eq:linearity}
            \support g \subseteq \Sigma_{\setminus J}^* \cdot \Sigma_J \cdot \Sigma_{\setminus J}^*.
        \end{align}
    \end{enumerate}
\end{lemma}
\noindent
\cite[Proposition~II.8]{Fliess:1981} corresponds to $J = \set{1, \dots, m}$.
\begin{proof}
    For every $n \in \N$, let $L_n$ be the set of words
    with exactly $n$ occurrences of letters from $\Sigma_J$.
    For instance, $L_0 = \Sigma_{\setminus J}^*$
    and $L_1 = \Sigma_{\setminus J}^* \cdot \Sigma_J \cdot \Sigma_{\setminus J}^*$.
    The implications ``$(4) \limplies (1)$'' and ``$(4) \limplies (2)$''
    follow immediately from~\cref{eq:additivity claim} and~\cref{eq:scalar multiple claim} (with $n = 1$), proved in the following two claims.
    \begin{claim*}
        For every $w \in L_1$ we have
        \begin{align}
            \label{eq:additivity claim}
            I_w(u_{\setminus J} + u_J + v_J)
                &= I_w(u_{\setminus J} + u_J) + I_w(u_{\setminus J} + v_J).
        \end{align}
    \end{claim*}
    \begin{claimproof}
        The claim follows by linearity of integration.
        Write $w = x \cdot a_j \cdot y \in L_1$ with $a_j \in \Sigma_J$ and $x, y \in L_0$.
        We proceed by induction on $x$.
        We will use the fact that $I_w(u)$ does not depend on $u_J$
        when $w$ does not contain letters from $\Sigma_J$.
        In the base case $w = a_j \cdot y$, we have
        $I_{a_j \cdot y}(u) = \int (u_j \cdot I_y(u))$,
        and we can apply linearity of integration:
        \begin{align*}
            &I_w(u_{\setminus J} + u_J + v_J)
            = \int ((u_j + v_j) \cdot I_y(u_{\setminus J} + u_J + v_J)) = \\
            &= \int ((u_j \cdot I_y(u_{\setminus J} + u_J + v_J)) + (v_j \cdot I_y(u_{\setminus J} + u_J + v_J))) = \\
            &= \int (u_j \cdot I_y(u_{\setminus J} + u_J)) + \int (v_j \cdot I_y(u_{\setminus J} + v_J)) = \\
            &= I_w(u_{\setminus J} + u_J) + I_w(u_{\setminus J} + v_J).
        \end{align*}
        For the inductive case, $w = a_h \cdot w'$
        for $a_h \in \Sigma_{\setminus J}$ and $w' \in L_1$:
        \begin{align*}
            &I_w(u_{\setminus J} + u_J + v_J)
            = \int (u_h \cdot I_{w'}(u_{\setminus J} + u_J + v_J)) = \\
            &= \int (u_h \cdot (I_{w'}(u_{\setminus J} + u_J) + I_{w'}(u_{\setminus J} + v_J))) = \\
            &= \int (u_h \cdot I_{w'}(u_{\setminus J} + u_J)) + \int (u_h \cdot I_{w'}(u_{\setminus J} + v_J)) = \\
            &= I_w(u_{\setminus J} + u_J) + I_w(u_{\setminus J} + v_J). \qedhere
        \end{align*}
    \end{claimproof}
    \begin{claim*}
        For every $n \in \N$ and $w \in L_n$,
        \begin{align}
            \label{eq:scalar multiple claim}
            F(u_{\setminus J} + \alpha \cdot u_J) = \alpha^n \cdot F(u_{\setminus J} + u_J),
            \quad \forall \alpha \in \R.
        \end{align}
    \end{claim*}
    \begin{claimproof}
        The proof is by induction on $n$, using linearity of integration as in the proof of~\cref{eq:additivity claim}.
    \end{claimproof}
    Regarding the implication ``$(1) \limplies (4)$'', assume that $F_g$ is additive on inputs $J$,
    and we need to show that $g$ is zero on $L_0 \cup L_2 \cup L_3 \cdots$.
    We can rewrite additivity~\cref{eq:additivity} as
    \begin{align*}
        \begin{aligned}
            &\sum_{w \in \Sigma^*} g_w \cdot (I_w(u + v_J)
            - I_w(u_{\setminus J} + u_J) -  I_w(u_{\setminus J} + v_J)) = 0.
        \end{aligned}
    \end{align*}
    Taking $u_J = v_J$ and using the fact that $\set{L_0, L_1, \dots}$ is a partition of $\Sigma^*$, we have
    \begin{align*}
        \begin{aligned}
            &\sum_{n = 0}^\infty \sum_{w \in L_n} g_w \cdot (I_w(u_{\setminus J} + 2 \cdot u_J)
            - 2 \cdot I_w(u_{\setminus J} + u_J)) = 0.
        \end{aligned}
    \end{align*}
    By~\cref{eq:scalar multiple claim} and since $u = u_{\setminus J} + u_J$ we can write
    \begin{align*}
        \begin{aligned}
            &\sum_{n = 0}^\infty \sum_{w \in L_n} g_w \cdot (2^n \cdot I_w(u)
            - 2 \cdot I_w(u)) = 0.
        \end{aligned}
    \end{align*}
    For every $w \in \Sigma^*$, let $n_w$ be the number of occurrences of $a_j$ in $w$.
    Let $g' \in \series \R \Sigma$ be the series obtained from $g$
    by setting $g'_w := (2^{n_w} - 2) \cdot g_w$ for all $w \in \Sigma^*$.
    With this notation, the last equation becomes $\sum_{w \in \Sigma^*} g'_w \cdot I_w(u) = 0$,
    but since $g'$ is zero on $L_1$, we have $\sum_{w \in L_0 \cup L_2 \cup \cdots} g'_w \cdot I_w(u) = 0$.
    By~\cref{lem:equality}, $g'$ must be zero on $L_0 \cup L_2 \cup \cdots$.
    Since $2^{n_w} - 2$ is never zero on the latter language,
    $g$ is zero on $L_0 \cup L_2 \cup \cdots$, as required.

    The same proof shows also ``$(2) \limplies (4)$'' by taking $\alpha = 2$ in~\cref{eq:proportionality}.
    Since linearity $(3)$ is equivalent to being both additive $(1)$ and proportional $(2)$,
    the proof is complete.
\end{proof}


The \emph{linearity problem} takes as input a polynomial system and a set of inputs $J \subseteq \set{1, \dots, m}$
and amounts to decide whether the corresponding causal functional is linear on inputs $J$.
Thanks to~\cref{lem:linearity}, linearity is reduced to the regular support inclusion query \cref{eq:linearity}
for the commutative regular language $\Sigma_{\setminus J}^* \cdot \Sigma_J \cdot \Sigma_{\setminus J}^*$.
By~\cref{thm:regular support inclusion}, we obtain the main result of this section.
\begin{theorem}
    The linearity problem for polynomial systems is decidable.
\end{theorem}

\begin{remark}
    In a similar way, one can decide whether a causal analytic functional
    is a quadratic form of its inputs.
\end{remark}

\subsection{Analytic systems}

For a tuple $u = \tuple{u_1, \dots, u_m} \in \D^m$
and an ordered set of indices $J = \set{j_1 < \cdots < i_n} \subseteq \set{1, \dots, m}$,
let $u_J := \tuple{u_{j_1}, \dots, u_{j_n}} \in \D^n$
and $u_{\setminus J} := u_{\set{1, \dots, m} \setminus J} \in \D^{m-n}$.
We write $\xi := \int u = \tuple{\int u_1, \dots, \int u_m} \in \D^m$.
For a multi-index $k = \tuple{k_1, \dots, k_n} \in \N^n$,
let $k! := k_1! \cdots k_n!$ and $\xi_J^k := \xi_{j_1}^{k_1} \cdots \xi_{j_n}^{k_n}$.
%
%
An analytic functional $F_g : \D^m \to \D$, with generating series $g \in \series \R \Sigma$,
is \emph{analytic in inputs $J$}
if it can be written as
\begin{align}
    \label{eq:analytic function}
    F_g(u) = \sum_{k \in \N^n} G_k(u_{\setminus J}) \cdot \frac {\xi_J^k} {k!},
\end{align}
where, for every $k \in \N^n$, $G_k$ is an analytic functional of the form $G_k = F_{g_k}$
with generating series $g_k \in \series \R {\Sigma_{\setminus J}}$.
In particular, when $J = \set{1, \dots, m}$, we say that $F_g$ is \emph{analytic} (in all inputs)
and it can be written as $F_g(u) = \sum_{k \in \N^m} G_k \cdot \frac {\xi^k} {k!}$ with $G_k \in \D$.
For instance $F_{a_1 a_2 + a_2 a_1}(u) = \int u_1 \int u_2 + \int u_2 \int u_1 = \xi_1 \cdot \xi_2$ is analytic,
but $F_{a_1 a_2}(u) = \int u_1 \int u_2$ is not.
We now provide a characterisation for this property.
%
%
\begin{restatable}[\protect{\cf~\cite[Proposition~II.9]{Fliess:1981}}]{lemma}{lemAnalyticCharacterisation}
    \label{lem:analytic characterisation}
    A causal analytic functional $F_g$ is analytic in $J$
    iff $g$ is commutative in $\Sigma_J$.
\end{restatable}
\noindent
As a particular case, when $J = \set{1, \dots, m}$, we recover \cite[Proposition~II.9]{Fliess:1981},
where commutative series are called \emph{exchangeable} (in French, \emph{échangeable}).

The \emph{analyticity problem} takes as input a polynomial system~\cref{eq:polynomial system}
and a set of indices $J \subseteq \set{1, \dots, m}$,
and asks whether its semantics is analytic in inputs $J$.
Thanks to~\cref{lem:analytic characterisation}, we have reduced this problem
to whether its generating series is commutative in a sub-alphabet $\Gamma$.
Thanks to~\cref{thm:commutativity}, we obtain the main result of this section.

\begin{theorem}
    The analyticity problem for polynomial systems is decidable.
\end{theorem}




\section{Future directions}

%
The formal approach based on formal power series that we have chosen
does not cover all possible aspects of polynomial systems.
In particular, we have left out aspects related to the \emph{time variable},
for instance properties such as being \emph{stationary} (invariant under temporal translation~\cite[Sec.~4(a)]{Fliess:1981})
or even \emph{time-invariant} (independent from the time variable, except via the inputs).
We do not know how to express these in the formal approach.
Nonetheless, the corresponding properties on generating series,
$\deriveright {a_0} g = 0$~\cite[Proposition~II.7]{Fliess:1981},
resp., $\support g \subseteq \Sigma_{\setminus \set{a_0}}^*$,
are decidable for shuffle-finite series using the same methods we have presented.

We mention an open problem which we find interesting.
In the \emph{bilinear immersion problem} we ask whether a given polynomial system is realised by a bilinear one.
Since bilinear systems correspond precisely to analytic functionals with rational generating series~\cite{Fliess:1981},
this reduces to the \emph{rationality problem} for shuffle-finite series,
which takes as input a shuffle-finite series and amounts to decide whether it is rational.
We do not know whether rationality of shuffle-finite series is decidable.



\addtolength{\textheight}{0cm}   


\bibliographystyle{plainurl}
\bibliography{literature}

\newpage
\appendix

\section{Full proofs}

\lemFliessHomomorphism*
\begin{proof}
    We provide a proof in order to illustrate the algebraic advantages of the formal approach.
    We first establish a basic first-step decomposition formula.
    \begin{claim}
        We have the decomposition
        \begin{align}
            \label{eq:derivative}
            (F_f(u))'
                &= \sum_{j = 0}^m u_j \cdot F_{\deriveleft {a_j} f}(u).
        \end{align}
    \end{claim}
    \begin{claimproof}
        By the fundamental relation between differentiation and integration,
        we have the following useful decomposition
        \begin{align}
            \label{eq:derivative of iterated integral}
            (I_w(u))' =
            \left\{
                \begin{array}{ll}
                    0
                        & \text{ if } w = \e, \\
                    u_j \cdot I_{w'}(u)
                        & \text{ if } w = a_j \cdot w'.
                \end{array}
            \right.
        \end{align}
        From this, we obtain~\cref{eq:derivative} as an immediate consequence:
        \begin{align*}
            (F_f(u))'
                &= \sum_{w \in \Sigma^*} g_w \cdot (I_w(u))' = \\
                &= \sum_{a_j \in \Sigma, w' \in \Sigma^*} g_{a_j \cdot w'} \cdot u_j \cdot I_{w'}(u) = \\
                &= \sum_{a_j \in \Sigma} u_j \cdot \sum_{w' \in \Sigma^*} g_{a_j \cdot w'} \cdot I_{w'}(u) = \\
                &= \sum_{a_j \in \Sigma} u_j \cdot \sum_{w' \in \Sigma^*} (\deriveleft {a_j} g)_{w'} \cdot I_{w'}(u) = \\
                &= \sum_{j = 0}^m F_{\deriveleft {a_j} f} \cdot u_j. \qedhere
        \end{align*}    
    \end{claimproof}
    Linearity~\cref{eq:scalar,eq:sum} is clear.
    We prove the product rule~\cref{eq:product}.
    Both sides are formal power series in $t$, so by coinduction it suffices to prove that they have the same constant term and derivative.
    For the constant term, we have
    \begin{align*}
        \coefficient {t^0} {(F_{f \shuffle g})}
            &= (f \shuffle g)_\e
            = f_\e \cdot g_\e, \text{ and } \\
        \coefficient {t^0} {(F_f \cdot F_g)}
            &= (\coefficient {t^0} {F_f}) \cdot (\coefficient {t^0} {F_g})
            = f_\e \cdot g_\e.
    \end{align*}
    For the derivative, by~\cref{eq:derivative} we have 
    \begin{align*}
        &(F_{f \shuffle g}(u))'
            = \sum_{j = 0}^m F_{\deriveleft {a_j} {(f \shuffle g)}}(u) \cdot u_j = \\
            &= \sum_{j = 0}^m F_{\deriveleft {a_j} f \shuffle g + f \shuffle \deriveleft {a_j} g}(u) \cdot u_j = \\
            &= \sum_{j = 0}^m (F_{\deriveleft {a_j} f \shuffle g}(u) + F_{f \shuffle \deriveleft {a_j} g}(u)) \cdot u_j = \\
            &= \sum_{j = 0}^m (F_{\deriveleft {a_j} f}(u) \cdot F_g(u) + F_f(u) \cdot F_{\deriveleft {a_j} g}(u)) \cdot u_j = \\
            &= \left(\sum_{j = 0}^m F_{\deriveleft {a_j} f}(u) \cdot u_j\right) \cdot F_g(u)
            + F_f(u) \cdot \left(\sum_{j = 0}^m F_{\deriveleft {a_j} g}(u) \cdot u_j\right) = \\
            &= (F_f(u))' \cdot F_g(u) + F_f(u) \cdot (F_g(u))' = \\
            &= (F_f(u) \cdot F_g(u))'. \qedhere
    \end{align*}
\end{proof}

We can combine the two homomorphisms~\cref{lem:Fliess homomorphism} and~\cref{lem:properties of semantics - shuffle automata} into a single homomorphism.
\begin{corollary}
    \label{cor:Fliess homomorphism}
    For every shuffle automaton $A$, let $\sem {\_}_A$ be the corresponding semantics.
    The composite map $F_{\sem {\_}_A} : \poly \R k \to (\D^m \to \D)$ is a homomorphism
    from the algebra of polynomials to the algebra of causal functionals:
    \begin{align}
        F_{\sem {c \cdot \alpha}}
            &= c \cdot F_{\sem \alpha},
            &&\forall c \in \R, \\
        F_{\sem {\alpha + \beta}}
            &= F_{\sem \alpha} + F_{\sem \beta}, \\
        F_{\sem {\alpha \cdot \beta}}
            &= F_{\sem \alpha} \cdot F_{\sem \beta}.
    \end{align}
\end{corollary}

\thmCoincidence*
\begin{proof}
    For completeness, we provide a proof of equivalence between the first and second point,
    in the syntax of shuffle automata.
    Fix an alphabet $\Sigma = \set{a_0, \dots, a_m}$.
    Consider a shuffle automaton $A = \tuple{\Sigma, X, \alpha, O, \Delta}$
    recognising the series $\sem \alpha$ with nonterminals $X = \set{X_1, \dots, X_k}$.
    We build a polynomial system $\SS = \tuple{m, k, x_0, p_0, \dots, p_m, q}$,
    where the initial state is $x_0 := O X = \tuple{O X_1, \dots, O X_k} \in \R^k$,
    for every $0 \leq j \leq m$, consider the tuple of polynomials
    $p_j := \tuple{\Delta_{a_j} X_1, \dots, \Delta_{a_j} X_k} \in \poly \R k^k$,
    and let the output polynomial be $q := \alpha \in \poly \R k$.
    We need to show $\sem \SS = F_{\sem \alpha}$.
    \begin{claim*}
        Fix an arbitrary input $u \in \D^m$
        and let $x \in \D^k$ be the corresponding unique power series solution of the polynomial system~\cref{eq:polynomial system}.
        We then have
        \begin{align}
            x = F_{\sem X} (u).
        \end{align}
    \end{claim*}
    \begin{claimproof}
        By uniqueness of solutions,
        it suffices to show that the tuple $F_{\sem {X}} = \tuple{F_{\sem {X_1}}, \dots, F_{\sem {X_k}}}$ satisfies
        the differential equation~\cref{eq:polynomial system} and the initial condition.
        Regarding the latter, $\coefficient {t^0} {(F_{\sem X} (u))} = \coefficient \e {\sem X} = O X = x(0)$.
        Regarding the former, we have 
        \begin{align*}
            &(F_{\sem X}(u))' =
                &&\text{(by~\cref{eq:derivative})} \\
            &= \sum_{j = 0}^m u_j \cdot F_{\deriveleft {a_j} \sem X}(u) =
                &&\text{(by~\cref{eq:shuffle:sem:derivation})} \\
            &= \sum_{j = 0}^m u_j \cdot F_{\sem {\Delta_{a_j} X}}(u) = 
                &&\text{(def.~of $p_j$)} \\
            &= \sum_{j = 0}^m u_j \cdot F_{\sem {p_j}}(u) = 
                &&\text{(by~\cref{cor:Fliess homomorphism})} \\
            &= \sum_{j = 0}^m u_j \cdot p_j(F_{\sem X})(u) =
                &&\text{(pointwise op.)} \\
            &= \sum_{j = 0}^m u_j \cdot p_j(F_{\sem X}(u)). \qedhere
        \end{align*}
    \end{claimproof}
    Thanks to the claim and~\cref{cor:Fliess homomorphism}, 
    for every input $u \in \D^m$ we have
    $F_{\sem \alpha} (u) = \alpha(F_{\sem X}(u)) = \alpha(x) = \sem \SS (u)$,
    as required.

    For the other direction, consider a polynomial system $\SS = \tuple{m, k, x_0, p_0, \dots, p_m, q}$.
    %
    We construct a shuffle automaton $A = \tuple{\Sigma, X, \alpha, O, \Delta}$
    over alphabet $\Sigma = \set{a_0, \dots, a_m}$,
    with nonterminals $X = \set{X_1, \dots, X_k}$,
    initial configuration $\alpha := q$,
    output function $O X := x(0)$,
    and transitions $\Delta_{a_j} X := p_j(X)$, for all $0 \leq j \leq m$.
    We need to show $\sem \SS = F_{\sem \alpha}$.
    %
    The proof is as in the previous direction.
\end{proof}

\thmClosureUnderCommutativeRegularSupportRestrictions*
\begin{proof}
    In the proof we have used the following claim.
    \claimRegularSupportRestrictions*
    We provide a complete proof of the claim.
    \begin{claimproof}
        Double restriction~\cref{eq:claim:restrict}
        and linearity, \cref{eq:claim:const,eq:claim:sum}, follow directly from the definitions.
        Regarding shuffle~\cref{eq:claim:shuffle},
        consider the following identity:
        \begin{align}
            \label{eq:claim:shuffle2}
            \restrict f x \shuffle \restrict g y
                &= \restrict {(\restrict f x \shuffle \restrict g y)} {x \cdot y},
                &&\forall x, y \in M.
        \end{align}
        %
        It follows from the fact that all words in the support of $u \shuffle v$ ($u, v \in \Sigma^*$)
        are mapped to the same monoid element $h(u) \cdot h(v)$, as shown in the following calculation:
        \begin{align*}
            &\restrict f x \shuffle \restrict g y
             = \left(\sum_{h(u) = x} f_u \cdot u\right) \shuffle \left(\sum_{h(v) = y} g_v \cdot v\right) =
                &&\text{} \\
            &= \sum_{h(u) = x, h(v) = y} f_u g_v \cdot (u \shuffle v) =
                &&\text{} \\
            &= \sum_{h(u) = x, h(v) = y} f_u g_v \cdot \restrict {(u \shuffle v)} {x \cdot y} =
                &&\text{} \\
            &= \restrictsmall {(\sum_{h(u) = x, h(v) = y} f_u g_v \cdot (u \shuffle v))} {x \cdot y} =
                &&\text{} \\
            &= \restrict {(\restrict f x \shuffle \restrict g y)} {x \cdot y}.
        \end{align*}
        A calculation suffices to establish~\cref{eq:claim:shuffle} from~\cref{eq:claim:shuffle2}:
        \begin{align*}
            &\restrict {(f \shuffle g)} m =
                &&\text{(by~\cref{eq:decomposition})} \\
            &= \restrictsmall {((\sum_{x \in M} \restrict f x) \shuffle (\sum_{y \in M} \restrict g y))} m =
                &&\text{} \\
            &= \sum_{x, y \in M} \restrict {(\restrict f x \shuffle \restrict g y)} m =
                &&\text{(by~\cref{eq:claim:shuffle2})} \\
            &= \sum_{x, y \in M} \restrict {\restrict {(\restrict f x \shuffle \restrict g y)} {x \cdot y}} m =
                &&\text{(by~\cref{eq:claim:restrict})}\\
            &= \sum_{m = x \cdot y} \restrict {(\restrict f x \shuffle \restrict g y)} {x \cdot y} =
                &&\text{(by~\cref{eq:claim:shuffle2})} \\
            &= \sum_{m = x \cdot y} (\restrict f x \shuffle \restrict g y),
        \end{align*}
        establishing~\cref{eq:claim:shuffle}.
%
        Regarding~\cref{eq:claim:derive}, take an arbitrary $w \in \Sigma^*$
        and we show that both sides agree on the coefficient of $w$.
        There are two cases.
        In the first case, suppose $h(a\cdot w) = h(a) \cdot h(w) = m$.
        %
        The coefficient of the \lhs~is
        $\coefficient w {(\deriveleft a {(\restrict f m)})} = \coefficient {a \cdot w} (\restrict f m) = f_{a \cdot w}$.
        And that of the \rhs, obtained in a unique way for $m' = h(w)$, is also $f_{a \cdot w}$.
        In the second case, suppose $h(a\cdot w) = h(a) \cdot h(w) \neq m$.
        As above, the coefficient of the \lhs~is $\coefficient {a \cdot w} (\restrict f m) = 0$,
        and that of the \rhs, obtained in a unique way for $m' = h(w)$, is also $0$.
    \end{claimproof}
\end{proof}

\lemAnalyticCharacterisation*
\begin{proof}



    For the ``if'' direction, assume that $g$ is commutative in $\Sigma_J$.
    %
    We use the fact that $g$ is commutative in $\Sigma_J$
    to group together terms with the same coefficient.
    Note that $\Sigma^*$ can be written as a disjoint union $\bigcup_{k \in \N^n} L_k$,
    where $L_k$ is the set of words with exactly $k$ occurrences of letters from $\Sigma_J$.
    In other words, $L_k$ is the support of $\Sigma_{\setminus J}^* \shuffle \Sigma_J^{\shuffle k}$,
    where $\Sigma_J^{\shuffle k} := a_{j_1}^{\shuffle k_1} \shuffle \cdots \shuffle a_{j_n}^{\shuffle k_n}$.
    Moreover, for every $w' \in \Sigma_{\setminus J}^*$,
    the series $g$ takes the same value on all words $w$ in the support of $w' \shuffle \Sigma_J^{\shuffle k}$,
    and this value is the coefficient of $w'$ in the series
    $g^{(k)} := \restrictsmall {(\deriveleft {a_{j_1}^{k_1} \cdots a_{j_n}^{k_n}} g)}{\Sigma_{\setminus J}^*} \in \series \R {\Sigma_{\setminus J}}$.
    The coefficient of $w$ in $w' \shuffle \Sigma_J^{\shuffle k}$ is $k!$,
    which we need to discount for.
    Finally, let $F_{\Sigma_J} := \tuple{F_{a_{j_1}}, \dots, F_{a_{j_n}}}$
    and we can thus write
    \begin{align*}
        &F_g(u)
        = \sum_{w \in \Sigma^*} g_w \cdot F_w(u) = \\
        &= \sum_{k \in \N^n} \sum_{w' \in \Sigma_{\setminus J}^*} \sum_{w \in \support{w' \shuffle \Sigma_J^{\shuffle k}}} g_w \cdot F_w(u) = \\
        &= \sum_{k \in \N^n} \sum_{w' \in \Sigma_{\setminus J}^*} g^{(k)}_{w'} \cdot \sum_{w \in \support{w' \shuffle \Sigma_J^{\shuffle k}}} F_w(u) = \\
        &= \sum_{k \in \N^n} \sum_{w' \in \Sigma_{\setminus J}^*}
            g^{(k)}_{w'} \cdot \frac 1 {k!} \cdot F_{w' \shuffle \Sigma_J^{\shuffle k}}(u) =
            &&\text{(by~\cref{eq:product})} \\
        &= \sum_{k \in \N^n} \sum_{w' \in \Sigma_{\setminus J}^*}
            g^{(k)}_{w'} \cdot \frac 1 {k!} \cdot F_{w'}(u) \cdot F_{\Sigma_J}(u)^k = \\ 
        &= \sum_{k \in \N^n} \sum_{w' \in \Sigma_{\setminus J}^*}
        g^{(k)}_{w'} \cdot F_{w'}(u) \cdot \frac {\xi_J^k} {k!} = \\
        &= \sum_{k \in \N^n} F_{g^{(k)}}(u_{\setminus J}) \cdot \frac {\xi_J^k} {k!}.
    \end{align*}
    Consequently, $F_g$ is analytic in $J$, as required.
    
    For the ``only if'' direction, assume that $F_g$ is analytic in $J$.
    By~\cref{eq:analytic function} we can write (where $g_k \in \series \R {\Sigma_{\setminus J}}$)
    \begin{align*}
        &F_g(u)
        = \sum_{k \in \N^n} F_{g_k}(u_{\setminus J}) \cdot \frac {\xi_J^k} {k!} = \\
        &= \sum_{k \in \N^n} F_{g_k}(u_{\setminus J}) \cdot \frac {F_{a_{j_1}}(u)^{k_1} \cdots F_{a_{j_n}}(u)^{k_n}} {k!} = \\
        &= \sum_{k \in \N^n} F_{g_k}(u_{\setminus J}) \cdot F_{a_{j_1}^{k_1} \shuffle \cdots \shuffle a_{j_n}^{k_n}}(u) = \\
        &= \sum_{k \in \N^n} F_{g_k \shuffle a_{j_1}^{k_1} \shuffle \cdots \shuffle a_{j_n}^{k_n}}(u) = \\
        &= \sum_{k \in \N^n} F_{h_k}(u) =
        \quad (\text{with } h_k := g_k \shuffle a_{j_1}^{k_1} \shuffle \cdots \shuffle a_{j_n}^{k_n}) \\
        &= F_{\sum_{k \in \N^n} h_k}(u).
    \end{align*}
    We have used the fact that the family of formal series $h_k$ is summable.
    Since $g_k$ is over $\Sigma_{\setminus J}$,
    these series are commutative in $\Sigma_J$, and so it is their sum.
    By~\cref{lem:equality}, we have $g = \sum_{k \in \N^n} h_k$,
    and thus $g$ is commutative in $\Sigma_J$.
\end{proof}

\end{document}